\documentclass[12pt,twoside]{article}

\usepackage{a4,enumerate}
\usepackage{amsmath, amsthm, amssymb}
\usepackage{tikz}

\usepackage{enumitem}

\makeatletter

\newtheorem{theorem}{Theorem}[section]
\newtheorem{corollary}[theorem]{Corollary}

\newtheorem{lemma}[theorem]{Lemma}
\theoremstyle{definition}
\newtheorem*{definition}{Definition}
\newtheorem{example}[theorem]{Example}

\newtheorem{remark}[theorem]{Remark}

\renewcommand{\@biblabel}[1]{#1.} 
\makeatother

\newcommand{\alg}{\operatorname{alg}}

\renewcommand{\Im}{\operatorname{Im}}

\renewcommand\phi{\varphi}

\renewcommand\epsilon{\varepsilon}
\newcommand\1{\mathbf{1}}

\providecommand{\dupa}[2]{\left\langle #1, #2 \right\rangle}
\providecommand{\scalprod}[2]{\left( #1\, \middle|\, #2 \right)}
\renewcommand{\sp}{\scalprod}
\providecommand{\abs}[1]{\left\lvert#1\right\rvert}
\providecommand{\norm}[1]{\left\lVert#1\right\rVert}
\providecommand{\set}[1]{\left\{ #1\right\}}
\newcommand\trace{\operatorname{tr}}
\newcommand\strace{\operatorname{str}}
\renewcommand{\P}{\mathbb{P}}

\DeclareMathOperator{\Tr}{Tr}
\DeclareMathOperator{\id}{id}
\DeclareMathOperator{\diam}{diam}

\newcommand{\st}{\mid}

\newcommand{\R}{\mathbb{R}\nonscript\hskip.03em}
\newcommand{\Z}{\mathbb{Z}\nonscript\hskip.03em}
\newcommand{\N}{\mathbb{N}\nonscript\hskip.03em}
\newcommand{\C}{\mathbb{C}\nonscript\hskip.03em}
\newcommand{\K}{\mathbb{K}\nonscript\hskip.03em}

\newcommand{\calA}{\mathcal{A}}
\newcommand{\calG}{\mathcal{G}}
\newcommand{\calV}{\mathcal{V}}
\newcommand{\calE}{\mathcal{E}}
\newcommand{\calB}{\mathcal{B}}
\newcommand{\calU}{\mathcal{U}}
\newcommand{\calS}{\mathcal{S}}
\newcommand{\calF}{\mathcal{F}}
\newcommand{\calP}{\mathcal{P}}
\newcommand{\calC}{\mathcal{C}}
\newcommand{\calZ}{\mathcal{Z}}

\newcommand{\cHg}{\mathcal{H}_\Gamma}

\newcommand{\from}{\colon}

\newcommand{\hm}[1]{\textbf{*}\leavevmode{\marginpar{\tiny%
$\hbox to 0mm{\hspace*{-0.5mm}$\leftarrow$\hss}%
\vcenter{\vrule depth 0.1mm height 0.1mm width \the\marginparwidth}%
\hbox to 0mm{\hss$\rightarrow$\hspace*{-0.5mm}}$\\\relax\raggedright #1}}}

\begin{document}

\title{Uniform existence of the integrated density of states on metric Cayley graphs}

\author{Felix Pogorzelski, Fabian Schwarzenberger and Christian Seifert}

\maketitle

\begin{abstract}
	Given an arbitrary, finitely generated, amenable group we consider ergodic Schr\"odinger operators on a metric Cayley graph with random potentials and random boundary conditions. We show that the normalised eigenvalue counting functions of finite volume parts converge uniformly. The integrated density of states (IDS) as the limit can be expressed by a Pastur-Shubin formula. The spectrum supports the corresponding measure and discontinuities correspond to the existence of compactly supported eigenfunctions. In this context, the present work generalises the hitherto known uniform IDS-approximation results for operators on the $d$-dimensional metric lattice to a very large class of geometries.  
\end{abstract}

MSC 2010: 47E05, 34L40, 47B80, 81Q10 \newline

Keywords: random Schr\"odinger operator, metric graph, quantum graph, integrated density of states

\section{Introduction}
The investigation of the spectrum of random Schr\"odinger operators has a long history in mathematical physics. 
One reason for this is that many solution properties of differential equations are encoded in the spectrum of the involved operator. 
The most prominent example for a result describing this connection is the RAGE-Theorem (see, e.g.\@ \cite[Theorem 4.1.19]{s2001}). 

In this paper we study spectral properties of random Schr\"odinger operators on
 Cayley graphs over finitely generated, amenable groups $\calG$. The graphs are defined 
as metric graphs, meaning that each edge is associated with an interval of a 
certain length. 
Hence, the operator acts on functions defined on the edges or the associated intervals. 
Selfadjointness is obtained by choosing appropriate 
boundary conditions at the vertices. 
We will only deal with local boundary conditions acting separately on each vertex $v$ and taking into account the 
functions linked with the edges connected to $v$. The involved randomness occurs in the 
choice of the potential, as well as in the boundary conditions. That way, we obtain a family $(H_{\omega})_{\omega \in \Omega}$ of selfadjoint operators, where $\Omega = (\Omega, \mathcal{A}, \mathbb{P})$ is a probability space endowed with some measure preserving, ergodic action by the amenable group $\mathcal{G}$ under consideration.  
Studying spectral properties of $(H_{\omega})$, it turns out that the relavant object is a function describing the distribution of the spectrum, the so-called \emph{spectral distribution function} (SDF). 
The SDF at point $\lambda\in \R$ is given via the trace of the spectral projection of the operator on the interval $(-\infty,\lambda]$, see also \eqref{eq:IDS}. In the special case where the operator in consideration is a Laplacian, on may intuitively think of the SDF to measure the number of electron energy levels below a given energy per unit volume.

As such a projection of an infinite dimensional operator is a rather 
abstract object, it is useful to think of a more constructive way to understand the spectral distribution. 
One approach in this direction is to consider a ``growing'' sequence $(Q_j)_{j\in\N}$ of finite sets exhausting the group $\mathcal{G}$. The approximating operators shall then be given as restrictions to these sets, such that one obtains a sequence $(H_\omega^{Q_j})_{j \in \mathbb{N}}$ 
of selfadjoint operators with discrete spectrum. For each $j$ one defines the 
eigenvalue counting function $n_\omega^{Q_j}$ as the function which maps 
$\lambda\in\R$ to the number of eigenvalues of $H_\omega^{Q_j}$, which do not 
exceed $\lambda$. Here one counts eigenvalues according to their 
multiplicities. 
Now the central issue is the behaviour of suitably normalised counting functions.
Precisely, the question arises whether there is some function
$N\from\R\to\R$ such that one can give meaning to the limit relation
\begin{equation}\label{def:ids}
N(\lambda):=\lim_{j\to\infty}\frac{1}{\abs{\calE_{Q_j}}}n_\omega^{Q_j}(\lambda).
\end{equation}
{\em If} the limit exists in a suitable sense, then we call $N$ the \emph{integrated density of states} 
(IDS). 
With this construction at hand several questions immediately arise:
\begin{itemize}
 \item Does the limit in \eqref{def:ids} exist? If so, in which topology?
 \item Is $N$ independent of the sequence $(Q_j)$ and/or of $\omega \in \Omega$?
 \item Does the IDS equal the SDF?
\end{itemize}
If the answer to the last question is positive, then one says that the \emph{Pastur-Shubin-trace formula} holds. Note that there are situations where the Pastur-Shubin-trace formula is not valid, c.f. \cite{as1993}.
Having defined the announced distribution functions, let us now discuss the physical relevance of these objects. 
As mentioned before, a considerable interest in the spectral type of selfadjoint operators has developed over the last decades. Results of this kind can be used to derive qualitative and quantitative properties of solutions to the corresponding difference or differential operator \cite{Teschl-09,Last-96}.  
To determine the spectral type, one may study particular properties of the SDF, such as (dis-)continuity or the asymptotic behaviour at the spectral edges. For instance, Lifshitz tails \cite{Nakao-77} and Wegner estimates \cite{Wegner-81} are crucial ingredients in the theory of Anderson localisation \cite{Anderson-58, s2001,PasturF-92}. 
Moreover, the investigation of continuity properties or the asymptotic behaviour of the IDS (or SDF) relies always on approximation techniques of the kind as described above.
In the light of that, the existence of the IDS and the validity of the 
Pastur-Shubin trace formula are valuable reference quantities when studying the
 spectral type of an operator, as well as solution properties of operator 
equations.

The first investigations concerning the IDS are given in the seminal papers of Pastur \cite{p1971} and Shubin \cite{s1979},
where the authors obtained weak convergence in a euclidean setting. 
Since then, a wide range of geometric situations and various classes of operators have been examined, see e.g.\@ \cite{as1993,lpv2004, s1990} 
for the continuous setting and e.g.\@ \cite{dlm2003, msy2003, my2002} for discrete geometries. 
It is well known that weak convergence of distribution functions only implies pointwise convergence at the continuity points of the limit function. 
On the other hand, the IDS may exhibit a large set of discontinuity points (which may be even dense in the spectrum), c.f.~\cite{kls2003,ls20031,v2005}. 
Therefore, in order to preserve essential spectral information in the approximation process (e.g.\@ spectral jumps indicating the existence of compactly 
supported eigenfunctions, cf.\@ \cite{lv2009}), the question of stronger kinds of convergence of the IDS arises. 
This motivates the main goal of our paper to verify uniform convergence for very general geometric settings.  
One particular approach 
is to use Banach space valued, almost-additive ergodic theorems. For operators on Delone sets modeling a quasicrystal structure, this method was first applied in \cite{ls2006}.   
For combinatorial Cayley graphs, convergence results for almost-additive functions have been used in \cite{lmv2008} for $\mathcal{G} =\mathbb{Z}^d$, as well as in \cite{lsv2011, lsv} for a finitely generated, amenable group $\mathcal{G}$ containing a monotile F{\o}lner sequence with respect to a symmetric grid. In \cite{ps2012}, the authors prove an almost-additive convergence result which is valid for all amenable groups. This latter theorem, in combination with the famous Lindenstrauss ergodic theorem for amenable groups (cf.\@ \cite{Lindenstrauss-01}) will be the major ingredient for our following elaborations, where we treat the general case of uniform IDS-approximation on all {\em metric} Cayley graphs induced by finitely generated, amenable groups. Precisely, we give positive answers to all of the above formulated questions and we prove uniformity for the limit relation given by \eqref{def:ids}.
 Hence, our major Theorem~\ref{thm:main} significantly extends the range of possible geometries. In particular, it allows for non-abelian structures. Thus, the present paper brings out a direct generalization of the results in \cite{glv2007,glv2008}, where the authors assumed that $\mathcal{G}=\mathbb{Z}^d$. Moreover, our main result is complementary to the elaborations in \cite{lppv2009} concerning an abstract CW complex graph model with group action, as well as to the ones in \cite{lsv2011,lsv,ps2012}, where the authors achieve uniform existence of the IDS in the context of combinatorial graphs over amenable groups.   

Let us briefly describe the content of this paper.
In section \ref{sec:amenable} we give some basic features on the geometric setting. 
Section \ref{sec:random_SO} describes the operator families in question. 
Restrictions to finite subgraphs are discussed in section \ref{sec:restriction}, where we also state the main results.
In section \ref{sec:ergodic} we apply an ergodic theorem obtained in \cite{ps2012} (see \cite{lsv2011} for a special case) 
to a sequence of spectral shift functions for an exhaustion of subgraphs.
These results are used in section \ref{sec:proof} to prove our main theorem. Section 7 is devoted to an application. 
We show uniform approximation of the IDS for the Heisenberg group (which is, in fact, a non-abelian group).

\section{Metric Cayley graphs over amenable groups}
\label{sec:amenable}
Let $\calG$ be a group, $\calS\subseteq \calG$ a finite but not 
necessarily symmetric set of generators and $\id\in\calG$ the unit element. 
We define the distance between two elements $g,h\in\calG$ to be the smallest 
number of elements in $\calS\cup\calS^{-1}$ one needs to turn $h$ into $g$ by left multiplication, i.\,e.\ 
\[
 d(g,h):=\begin{cases}\min\{k\in \N \st  \exists s_1,\dots,s_k\in \calS\cup\calS^{-1} 
\text{ with } s_1 \cdots s_kh=g\} &\text{if } g\neq h\\ 0 &\text{otherwise}.  \end{cases}
\]
We denote the set of all finite subsets of $\calG$ by $\calF$. The diameter of a set $Q\in \calF$ is given by 
$\diam Q:=\max\{d(g,h)\st g,h\in Q\}$. For a subset $Q\subseteq\calG$ 
and $g\in \calG$ we set $d(g,Q):=\min\{d(g,h)\st h\in Q \}$.
Given $R\in \N$ and $Q\in\calF$ we define $\partial^R Q=\{ g\in \calG\st g\in Q, d(g,G\setminus Q)\leq R \text{ or } g\notin Q, d(g, Q)\leq R\}$.
We assume that $\calG$ is amenable, i.\,e., there exists a sequence $(Q_j)_{j\in \N}$ of elements in $\calF$ such that
\[
  \lim_{j\to\infty}\frac{|\calS Q_j\setminus Q_j|}{|Q_j|} = 0.
\]
Such a sequence $(Q_j)$ is called \emph{F\o lner sequence}. 
A F\o lner sequence is said to be \emph{tempered} if there exists $C>0$ such that
\[
 \left\vert \bigcup_{k=1}^{j-1} Q_k^{-1}Q_j \right\vert \leq C\vert Q_j\vert
\]
 holds for all $j\in \N$. 

\begin{remark}\label{remark_Folner}
\begin{itemize}
\item [(a)] It is easy to see that for a F\o lner sequence $(Q_j)_{j\in \mathbb N}$, the limit relation 
\[
 \lim_{j\to \infty} \frac{|\partial^R Q_j|}{|Q_j|} = 0
\]
holds true for all $R\in \mathbb N$, c.f.\ \cite{lsv2011}.
\item [(b)] Note that each F\o lner sequence has a tempered subsequence, 
c.f.~\cite{Lindenstrauss-01}. Therefore there exists a tempered F\o{}lner 
sequence in each amenable group.
\end{itemize}
\end{remark}
For a given group $\calG$ and a finite set of generators $\calS$ we denote the induced (directed) metric Cayley graph by 
$\Gamma= \Gamma(\calG,\calS) = (\calV,\calE,\gamma)$, i.e., $\calV = \calG$ is the vertex set, 
$\calE$ the set of edges and $\gamma=(\gamma_0,\gamma_1)\from \calE\to \calV\times\calV$ associates to each edge 
$e\in\calE$ the starting vertex $\gamma_0(e)$ and the end vertex $\gamma_1(e)$. 
There will be an edge $e$ from $v$ to $w$ if there exists $s\in \calS$ such that $w=sv$. Such an edge $e$ will then be said to be generated by $s$.
Every edge $e\in\calE$ generated by $s\in\calS$ will be identified with the interval $[0,l_s]$, where $l_s\in (0,\infty)$ for all $s\in\calS$.

\begin{example}\label{ex:Z2}
 Let $\calG=\Z^2$ and set
\[
\calS_1=\{(1,0),(0,1)\}\quad\text{and}\quad 
\calS_2=\{(0,0),(1,1),(1,0),(-1,0)\}.
\]
 Then $\calS_1$ and $\calS_2$ are generating systems for $\calG$. We denote the corresponding metric Cayley graphs by $\Gamma_1=\Gamma_1(\calG,\calS_1)$ and $\Gamma_2=\Gamma_2(\calG,\calS_2)$. Note that while $\Gamma_1$ is the usual graph of $\Z^2$
with standard edges, $\Gamma_2$ contains multiple edges as well as loops.

\begin{figure}[h] \centering \label{abb:dsf}
\begin{tikzpicture}[ scale=0.8]
{
{\color{black}
\foreach \y in {1,2,3,4}{
	\foreach \x in {1,2,3,4,5,6}
		\filldraw (\x,\y) circle (3pt);
}
\foreach \y in {1,2,3,4}{
\draw[-latex] (1,\y) --(1.9,\y);
\draw[-latex] (2,\y) --(2.9,\y);
\draw[-latex] (3,\y) --(3.9,\y);
\draw[-latex] (4,\y) --(4.9,\y);
\draw[-latex] (5,\y) --(5.9,\y);
}
\foreach \y in {1,2,3,4,5,6}{
\draw[-latex] (\y,1) --(\y,1.9);
\draw[-latex] (\y,2) --(\y,2.9);
\draw[-latex] (\y,3) --(\y,3.9);

}
}
}
{\color{black}
\foreach \y in {1,2,3,4}{
	\foreach \x in {9,10,11,12,13,14}{
		\filldraw (\x,\y) circle (3pt);
		\draw[-latex,rounded corners=2.8pt] (\x,\y +0.081) --++ (-0.2,0.4) --++ (0.2,0.2) --++ (0.2,-0.2) --++ (-0.2,-0.4);
}
}
%

\foreach \y in {1,2,3,4}{
\path[-latex]  (9,\y) edge [bend left]  (9.9,\y);
\path[-latex]  (10,\y) edge [bend left]  (10.9,\y);
\path[-latex]  (11,\y) edge [bend left]  (11.9,\y);
\path[-latex]  (12,\y) edge [bend left]  (12.9,\y);
\path[-latex]  (13,\y) edge [bend left]  (13.9,\y);
\path[-latex]  (10,\y) edge [bend left]  (9.1,\y);
\path[-latex]  (11,\y) edge [bend left]  (10.1,\y);
\path[-latex]  (12,\y) edge [bend left]  (11.1,\y);
\path[-latex]  (13,\y) edge [bend left]  (12.1,\y);
\path[-latex]  (14,\y) edge [bend left]  (13.1,\y);

}
\foreach \x / \xt in {1/2,2/3,3/4,4/5,5/6}{
\foreach \y / \yt in {1/2,2/3,3/4}{
\draw[-latex] (\x +8,\y) --(\xt +8-0.08,\yt-0.08);
}}
}
\end{tikzpicture}
\caption{Illustration of $\Gamma_1(\calG,\calS_1)$ and $\Gamma_2(\calG,\calS_2)$ from Example \ref{ex:Z2} }
\end{figure}
\end{example}

\section{Random Schr\"odinger Operators on graphs}
\label{sec:random_SO}

All function spaces appearing in this paper will be $\K$-valued, where 
$\K\in\set{\R,\C}$.

For $s\in\calS$ let $\calB_s\subseteq L^\infty(0,l_s)$ be a finite subset.
For $e\in \calE$ generated by $s\in\calS$ let $V_e\in \calB_s$.
In the Hilbert space
\[
\cHg := \bigoplus_{v\in \calV} \bigoplus_{s\in\calS} L^2(0,l_s)
\]
we define the maximal operator
\begin{align*}
	D(\hat{H}) & := \bigoplus_{v\in \calV} \bigoplus_{s\in\calS} W^{2,2}(0,l_s) ,\\
	(\hat{H}f)_e & := -f_e'' + V_ef_e \quad (e\in \calE).
\end{align*}

In order to obtain selfadjoint realisations we need to impose boundary conditions at the vertices. We will not consider the most general boundary conditions, but rather restrict ourselves to so-called local boundary conditions.

For $v\in \calV$, the collections
\[
 \calE_{v,j}:=\set{e\in \calE \st \gamma_j(e)=v}\qquad (j=0,1)
\]
describe the sets of all edges starting or ending at $v$,
respectively, and
\[
 \calE_{v}:=\big(\calE_{v,0}\times\{0\}\big)\cup\big(\calE_{v,1}\times\{1\}\big)
\]
encodes all edges connected with $v$ (where loops are counted twice).

For $f\in D(\hat{H})$ and $v\in\cal V$ we define the \emph{trace mapping} (or \emph{boundary value
mapping}) $\trace_v f \in \K^{\calE_v}$ by
\[
 (\trace_v f) (e,j) := \begin{cases} f_e(0) & (e,j)\in \calE_v, j=0,\\
                        f_e(l_s) & (e,j)\in \calE_v, j=1, e\,\text{generated by}\, s.
                       \end{cases}
\]
Furthermore, define the \emph{signed trace}
$\strace_v f' \in \K^{\calE_v}$ by 
\[
 (\strace_v f') (e,j):= \begin{cases} f_e'(0) & (e,j)\in \calE_v, j=0,\\
                        -f_e'(l_s) & (e,j)\in \calE_v, j=1, e\,\text{generated by}\, s.
                       \end{cases}
\]

\begin{remark}\begin{itemize}
               \item [(a)] Note that $W^{2,2}(0,l)\subseteq C^1[0,l]$ by standard Sobolev arguments and hence for $f\in D(\hat{H})$, the vectors $\trace_vf$ and $\strace_v f'$ are well-defined ($v\in \calV$).
	       \item [(b)] The definition of the signed trace implies that the 
orientation of the edges plays a minor role. In particular, only the 
boundary conditions (see the definition below) take into account the
 direction of the edges.
              \end{itemize}	
\end{remark}

\begin{definition}[local boundary conditions]
	Let $v\in \calV$. Local boundary conditions at $v$ are encoded in a subspace $U_v\subseteq \K^{\calE_v}\oplus\K^{\calE_v}$ with $\dim U_v = \abs{\calE_v}$ such that
	\[\sp{f_1'}{f_2} - \sp{f_1}{f_2'} = 0 \quad((f_1,f_1'),(f_2,f_2')\in U_v),\]
	where $\sp{\cdot}{\cdot}$ denotes the usual inner product in $\K^{\calE_v}$.
	We say that $f\in D(\hat{H})$ satisfies the \emph{local boundary condition $U_v$ at $v\in V$}, if $(\trace_vf,\strace_v f')\in U_v$.
	\emph{Local boundary conditions} are a family $U:= (U_v)_{v\in\calV}$ of local boundary conditions at each vertex $v\in \calV$.
\end{definition}

For a local boundary condition $U$ the operator
\begin{align*}
	D(H) & := \set{f\in D(\hat{H}) \st (\trace_v f, \strace_v f')\in  U_v 
\quad (v\in \calV)}, \\
	(Hf)_e & := (\hat{H}f)_e = -f_e'' + V_e f_e \quad(e\in \calE)
\end{align*}
is selfadjoint, cf. \cite{ks99,kps08,kuc04,glv2007}.

\begin{example}
	(a) Dirichlet boundary conditions. Let $U_v^D:= \set{0}^{\calE_v}\oplus \K^{\calE_v}$. Then $U_v^D$ encodes Dirichlet boundary conditions at $v$, since $\trace_v f = 0$ ($f\in D(H)$).

	(b) Neumann boundary conditions. Let $U_v^N:= \K^{\calE_v}\oplus\set{0}^{\calE_v}$. Then $U_v^N$ encodes Neumann boundary conditions at $v$, since $\strace_v f' = 0$ ($f\in D(H)$).
\end{example}

\begin{example}[Dirichlet-Laplacian]
	Let $V_e = 0$ for all $e\in \calE$. Then the operator $H$ with Dirichlet 
boundary conditions $(U_v^D)_{v\in\calV}$ is called \emph{Dirichlet 
Laplacian} and is denoted by $-\Delta_D$. We have
	\begin{align*}
		D(-\Delta_D) & = \bigoplus_{v\in \calV} \bigoplus_{s\in\calS} 
W_{0}^{1,2}(0,l_s)\cap W^{2,2}(0,l_s),\\
		(-\Delta_D f)_e & = -f_e'' \quad(e\in\calE).
	\end{align*}
\end{example}

Now, we introduce randomness in the choice of potentials and boundary conditions.

Note that $\calG$ acts on $\Gamma$ in the following way: For $e\in \calE$ and $g\in \calG$ there is also a unique edge $e\circ g\in \calE$ connecting $\gamma_0(e)g^{-1}$ and $\gamma_1(e)g^{-1}$. Shorthand, we can therefore write
\[\gamma(e\circ g) = (\gamma_0(e)g^{-1},\gamma_1(e)g^{-1}).\]

Let $(\Omega,\calA,\P)$ be a probability space and let $\calG$ act ergodically on $(\Omega,\calA,\P)$, i.\,e., if
$\alpha\from \calG\times\Omega\to \Omega$ is the group action on $\Omega$, then every subset of $\Omega$ which is invariant under $(\alpha_g)_{g\in\calG}$ has measure either zero or one. Additionally we want $\alpha$ to act measure preservingly, i.\,e.\ $\P(A)=\P(\alpha_g(A))$ for all $g\in \calG$ and all $A\in\calA$.

A random potential is a map $V\from \Omega\to 
\prod_{v\in\calV}\prod_{s\in\calS} \calB_s$ satisfying
\begin{equation}\label{eq:potential}
V(\alpha_g(\omega))_{e\circ g} = V(\omega)_e \quad (g\in \calG, e\in \calE).
\end{equation}

Since a Cayley graph is regular (in the sense that every vertex has the same degree and for two vertices there exists a bijective mapping between the adjacent edges at these vertices), we can choose local boundary conditions $U_{\id}$ at $\id\in \calV$ and then shift these boundary conditions to an arbitrary $v\in\calV$ to obtain a local boundary condition at $v$.
Hence we can choose random boundary conditions in the following way:

Let $\calU$ be a finite set of local boundary conditions at $\id$. A random boundary condition is a map $U\from \Omega\to \prod_{v\in \calV} \calU$ satisfying
\begin{equation}\label{eq:bcond}
 U(\alpha_g(\omega))_{v} = U(\omega)_{vg} \quad(g\in\calG, v\in\calV).
\end{equation}

The family of random Schr\"odinger operators $(H_\omega)_{\omega\in\Omega}$ on $\cHg$ is defined by
\begin{align}
	D(H_\omega) & := \set{f\in D(\hat{H}) \st (\trace_v f,\strace_v f')\in 
U(\omega)_v \quad(v\in\calV)}, \label{def:H_omega1} \\
	(H_\omega f)_e & := -f_e'' + V(\omega)_e f_e \quad(e\in\calE), \label{def:H_omega2}
\end{align}
for $\omega\in\Omega$. For each $\omega\in \Omega$, $H_\omega$ is selfadjoint and semibounded from below. More precisely, there is $C\geq0$ such that $H_\omega + C\geq 0$ for all $\omega\in\Omega$.

\section{Restrictions to finite subsets}
\label{sec:restriction}

Let $Q\subseteq \calG$ be a finite subset. The associated subgraph $\Gamma_Q=(\calV_Q,\calE_Q,\gamma_Q)$ of $\Gamma$ is defined as follows:
\[\calE_Q := \bigcup_{v\in Q} \calE_{v,0},\quad	\calV_Q := Q\cup \calS Q, \quad\gamma_Q := \gamma|_{\calE_Q}.\]

We also define inner vertices $\calV_Q^i$ and boundary vertices $\calV_Q^\partial$ by
\[	\calV_Q^i := \set{v\in \calV_Q \st \calE_{v,0}\cup \calE_{v,1} \subseteq \calE_Q}, \quad \calV_Q^\partial := \calV_Q\setminus \calV_Q^i,\]

and accordingly inner edges $\calE_Q^i$ and boundary edges $\calE_Q^\partial$ by
\[	\calE_Q^i  := \set{e\in \calE_Q \st \gamma_0(e),\gamma_1(e)\in \calV_Q^i},\quad 
	\calE_Q^\partial  := \calE_Q\setminus \calE_Q^i.\]

We define the restriction $H_\omega^Q$ of $H_\omega$ to $\Gamma_Q$ on 
\[\mathcal{H}_{\Gamma_Q} = \bigoplus_{v\in Q} \bigoplus_{s\in\calS} L^2(0,l_s)\]
by
\begin{align*}
	D(H_\omega^Q) & := \Big\{f\in \bigoplus_{v\in Q} \bigoplus_{s\in\calS} W^{2,2}(0,l_s) \st (\trace_v f,\strace_vf') \in U(\omega)_v \quad(v\in \calV_Q^i),\\
			& \hspace{12.9em} (\trace_v f,\strace_vf') \in U_v^D \quad(v\in \calV_Q^\partial)\Big\},\\
	(H_\omega^Q f)_e & := -f_e'' + V(\omega)_e f_e \quad(e\in \calE_Q).
\end{align*}

This operator is again selfadjoint and semibounded from below.
Furthermore, $H_\omega^Q$ has purely discrete spectrum; cf.\ \cite[Theorem 18]{kuc04}.

Let us enumerate the eigenvalues $(\lambda_n(H_\omega^Q))_{n\in\N}$ as an increasing sequence, counting their multiplicities. The \emph{eigenvalue counting function} $n_\omega^Q\from \R\to \N_0$ is defined by
\[n_\omega^Q(\lambda) := |\set{n\in\N \st \lambda_n(H_\omega^Q)\leq \lambda}|= \Tr \1_{(-\infty,\lambda]}(H_\omega^Q).\]
Then $n_\omega^Q$ is monotonically increasing and right continuous, i.e. a distribution function. 
The volume-scaled version of $n_\omega^Q$ will be denoted by $N_\omega^Q$, i.\,e.,
\[N_\omega^Q(\lambda):= \frac{1}{\abs{\calE_Q}}n_\omega^Q(\lambda) \quad(\lambda\in \R).\]
It is associated to a pure point measure $\mu_\omega^Q$.
Note that $\abs{\calE_Q} = \abs{\calS}\abs{Q}$.

We now state the main theorem of this paper.

%

\begin{theorem}\label{thm:main}
	Let $(Q_j)_{j\in\N}$ be a tempered F\o lner sequence in $\calG$. Then there is $N\from \R\to\R$ monotone increasing and right continuous (i.\,e.\ a distribution function), such that
	\[\lim_{j\to\infty} \norm{N_\omega^{Q_j}-N}_\infty = 0\]
	for $\P$-a.\,a.\ $\omega\in\Omega$. In particular, $N_\omega^{Q_j}\to N$ pointwise for $\P$-a.\,a.\ $\omega\in \Omega$.

	Furthermore, for $\lambda\in\R$ and $Q\subseteq \calG$ finite
	\begin{equation}\label{eq:IDS}
		N(\lambda) = \frac{1}{\abs{\calE_Q}} \int_\Omega \Tr\left(\1_{\calE_Q} \1_{(-\infty,\lambda]}(H_\omega)\right)\, d\P(\omega).
	\end{equation}
	Note that $N(\lambda)$ does not depend on the choice of $Q$.
\end{theorem}

The distribution function $N$ is called the \emph{integrated density of states (IDS)}. Let $\mu$ be the corresponding measure. Theorem \ref{thm:main} states that the IDS is the uniform limit of the normalised eigenvalue counting functions on finite subgraphs and can be expressed by a Pastur-Shubin trace formula in \eqref{eq:IDS}. The operator $\Tr$ denotes the usual trace in $L^2$.

By ergodicity of $(H_\omega)_{\omega\in\Omega}$ we obtain the following Theorem, which is an analogue of \cite[Theorem 5]{glv2007}. For the proof we may apply the general framework of \cite[Theorem 5.1]{lpv2007}.

\begin{theorem}
	There exist subsets $\Sigma, \Sigma_{\rm pp}, \Sigma_{\rm sc}, \Sigma_{\rm ac}, \Sigma_{\rm disc}, \Sigma_{\rm ess} \subseteq \R$ and $\Omega'\subseteq \Omega$ with $\P(\Omega') = 1$ such that $\sigma(H_\omega) = \Sigma$ and $\sigma_\bullet(H_\omega) = \Sigma_{\bullet}$ for all the spectral types $\bullet\in \set{\rm pp,sc,ac,disc,ess}$ and all $\omega\in \Omega'$.
\end{theorem}

As a consequence, we can relate the measure $\mu$ with the $\P$-a.\,s.\ spectrum $\Sigma$ of $(H_\omega)$, cf.\ \cite{lpv2007, glv2007}.

\begin{corollary}
\label{cor:topological_spt}
	$\Sigma$ is the topological support of $\mu$.
\end{corollary}

Denote by
\[D := \set{f\in \bigoplus_{v\in \calV} \bigoplus_{s\in\calS} L^2(0,l_s) \st \exists \, \mbox{$\calE'\subseteq \calE$ finite}: f_e = 0 \quad(e\in\calE\setminus\calE')}\]
the set of compactly supported $L^2$-functions on $\Gamma$.

\begin{corollary}
\label{cor:compactly_spt_ef}
	Let 
	\[\Sigma_{\rm comp} := \set{\lambda\in \R \st \mbox{for $\P$-a.\,a.\ 
$\omega\in\Omega\,\exists\,f_\omega\in D(H_\omega)\cap D$}: H_\omega 
f_\omega = \lambda f_\omega}.\]
	Then 
	\[\Sigma_{\rm comp} = \set{\lambda\in\R \st \mu(\set{\lambda})>0}.\]
\end{corollary}

\begin{remark}
	(a) The set $\set{\lambda\in\R \st \mu(\set{\lambda})>0}$ is the set of atoms of $\mu$ and equals the set of discontinuities of the IDS.

	(b) The proof of Corollary \ref{cor:compactly_spt_ef} follows the lines of \cite[Proof of Corollary 7]{glv2007}, applying Theorem \ref{thm:main}.
\end{remark}

\section{Convergence of spectral shift functions}
\label{sec:ergodic}

The next aim is the application of a Banach-space valued ergodic theorem given in \cite{lsv2011}. Therefore it is necessary to prove certain properties of the spectral shift functions.
Before this we introduce the notion concerning the colouring of the Cayley graph $\Gamma=(\calV,\calE,\gamma)$ associated to a given group $\calG$ with finite set of generators $\calS$.

Let $\calZ$ be an arbitrary finite set. A map $\calC: \calV \to \calZ$ is called a \emph{colouring} 
of $\Gamma$ and a map $P:D(P)\to \calZ$, where $D(P)\in \calF$, a \emph{pattern}. Note that, as before, 
$\calF$ denotes the set of all finite subsets of $\calG$.
We write $\calP$ for the set of all patterns and for given $Q\in \calF$ we define the set 
$\calP(Q):= \{P\in\calP \st D(P)=Q \}$.
Given a pattern $P$ and a set $Q\subseteq D(P)$ the \emph{restriction of $P$ on $Q$} is the map 
$P|_Q: Q\to \calZ$ with $P|_Q(g)=P(g)$ for all $g\in Q$. Equivalently, the \emph{restriction of a colouring} $\calC$ to a finite set 
$Q\in\calF$ is given by $\calC|_Q:Q\to \calZ, \calC|_Q(g)=\calC(g)$ for all $ g\in Q$. For $P\in \calP$ and $x\in \calG$ 
the \emph{translation of $P$ by $x$} is defined by $Px:D(P)x\to \calZ$, $(Px)(g)=P(gx^{-1})$. 
We say that two patterns $P,P'\in\calP$ are \emph{equivalent} (and write $P\sim P'$) if there exists $x\in \calG$ 
with $D(P)x=D(P')$ and $(Px)(g)=P'(g)$ for all $g\in D(P')$. The induced quotient set is denoted by $\tilde\calP$
and the equivalence class for given $P\in\calP$ by $\tilde P\in \tilde \calP$. 
For given patterns $P_1,P_2\in \calP$ we set $\sharp_{P_1}({P_2})$ to be the number of occurrences of $P_1$ in $P_2$, i.\,e.\
\[
 \sharp_{P_1}({P_2}):= \vert\{ P \in \calP \st P\sim P_1, D(P)\subseteq D(P_2)  \}\vert.
\]
\begin{definition}
 A function $b:\calF\to [0,\infty)$ is called \emph{boundary term} if 
\begin{itemize}
 \item $b(Q)=b(Qx)$ for all $Q\in \calF$ and $x\in \calG$,
 \item $\lim_{j\to\infty}|Q_j|^{-1}b(Q_j)=0$ for any F\o{}lner sequence $(Q_j)$,
 \item $\exists C>0$ such that $|Q|^{-1}b(Q)\leq C$ for all $Q\in \calF$ and
 \item one has for all $Q,Q'\in \calF$
\[
 b(Q\cap Q')\leq b(Q)+b(Q'),\; b(Q\cup Q')\leq b(Q)+b(Q'),\; b(Q\setminus Q') \leq b(Q)+b(Q').
\]
\end{itemize}
\end{definition}
\begin{definition}
 Let $(X,\Vert\cdot\Vert)$ be a Banach space and a function $ F: \calF \to X$ be given. $F$ is called 
\begin{enumerate}
 \item \emph{almost additive} if there exists a boundary term $b:\calF\to[0,\infty)$ such that for any pairwise disjoint 
  subsets $Q_j$, $j=1,\dots,k$
 \[
   \left\Vert F(Q) - \sum_{j=1}^k  F( Q_j) \right\Vert \leq \sum_{j=1}^k b( Q_j)
 \]
holds, where $Q=\bigcup_{j=1}^k Q_j$;
 \item \emph{$\calC$-invariant} if $F(Q)=F(Qx)$ for all $x\in \calG$ and all $Q\in\calF$ with $\calC|_Q\sim \calC|_{Qx}$.
\end{enumerate}
\end{definition}
For a given almost additive and $\calC$-invariant function $F:\calF\to X$ we define a function $\tilde F:\tilde\calP\to X$ by setting
\[
 \tilde F(\tilde P)=\begin{cases}
                     F(Q) &\mbox{if } \exists\ Q\in \calF \mbox{ such that } \tilde {\calC|_Q} = \tilde P, \\0 & \mbox{otherwise.}
                    \end{cases}
\]
Note that this is well-defined by $\calC$-invariance of $F$.
\begin{theorem}\label{theorem_erg}
 Let $\calG$ be an amenable group generated by a finite set $\calS$, $\Gamma=(\calV,\calE,\gamma)$ the associated Cayley graph, 
 $\calZ$ a finite set and $\calC: \calV\to \calZ$ be an arbitrary colouring. 
 Let $(Q_j)$ be a F\o{}lner sequence in $\calG$ and assume that the frequencies $\nu_P:=\lim_{j\to\infty}|Q_j|^{-1}\sharp_P(\calC|_{Q_j})$ 
 exist for all patterns $P\in \calP$.
Furthermore let $(X,\Vert\cdot\Vert)$ be a Banach-space and $F:\calF\to X$ an almost additive and $\calC$-invariant mapping. Then the limit
\[
 \lim_{j\to\infty}\frac{F(Q_j)}{|Q_j|} 
\]
 exists in the topology of $X$.
\end{theorem}
See \cite{ps2012} for a proof of Theorem \ref{theorem_erg} and \cite{lsv2011, lmv2008} for earlier versions. 
Now we show that in our situation the assumptions of the theorem are fulfilled almost surely. 
For $\omega\in\Omega$ define the map $\calC_\omega:\calV\to \calZ$ by
\begin{equation}\label{def_cC_omega}
 \calC_\omega(v) := \left( (V(\omega)_{e})_{e\in\calE_{v,0}} ,U(\omega)_v \right)
\quad\mbox{where}\quad 
\calZ:=(\oplus_{s\in \calS}\calB_s)\times \calU .
\end{equation}
To show the existence of the frequencies $\nu_P$ we need the following theorem, which is a special case of the Lindenstrauss' pointwise ergodic theorem in \cite{Lindenstrauss-01}.
\begin{theorem}\label{theorem_linde}
  Let $\calG$ act from the left on a measure space $(\Omega,\calA,\P)$ by an ergodic and measure preserving transformation $\alpha$ and let $(Q_j)$ be a tempered F\o{}lner sequence. Then for any $ f\in L^1(\P)$
  \[
  \lim_{j\rightarrow\infty}\frac{1}{\vert Q_j \vert}\sum_{g\in Q_j}f(\alpha_g (\omega) ) =\int_\Omega f(\omega) d \P(\omega)
  \]
  holds for $\P$-a.\,a.\ $\omega\in\Omega$.
\end{theorem}
\begin{lemma}
 Let $(Q_j)$ be a tempered F\o lner sequence and $\calC_\omega$ and $\calZ$ be given as in \eqref{def_cC_omega} for all $\omega \in \Omega$. Then there exists a set $\tilde\Omega\subseteq \Omega$ with $\P(\tilde\Omega)=1$, such that for each $P\in\bigcup_{j\in\N}\calP(Q_j)$ and $\omega \in \tilde\Omega$ the limit
\[
 \nu_P = \lim_{j\to\infty}\frac{\sharp_P(\calC_\omega|_{Q_j})}{|Q_j|}
\]
exists and is independent of $\omega\in\tilde\Omega$.
\end{lemma}
\begin{proof}
 Let $P:Q\to \calZ$ be given for some $Q\in \calF$ with $\diam Q = R$. W.\,l.\,o.\,g.\ we may assume that $\id\in Q$, which is possible since we are interested in counting translates of $P$.
Obviously the following inequalities hold
\[
\sum_{g\in Q_j\setminus \partial^R Q_j} \1_{A(\omega)}(g)
\leq
\sharp_P(\calC_\omega|_{Q_j})
\leq 
\sum_{g\in Q_j} \1_{A(\omega)}(g),
\]
where $A(\omega):= \{g\in \calG\st P(v)=\calC_\omega(vg) \text{ for all }v\in Q\}$.
By the properties of $U$ and $V$ given in \eqref{eq:potential} and \eqref{eq:bcond} we have for given $g,v\in \calG$, $\omega\in \Omega$
\[
 \calC_\omega(vg)=\left( (V(\omega)_e)_{e\in \calE_{vg,0}}, U(\omega)_{vg} \right)
=\left( (V(\alpha_g(\omega))_e)_{e\in \calE_{v,0}}, U(\alpha_g(\omega))_{v} \right)
=\calC_{\alpha_g(\omega)}(v).
\]
Therefore,
\[
 \1_{A(\omega)}(g) = \1_{\{g\in \calG\st P(v)=\calC_{\alpha_g(\omega)}(v) \text{ for all }v\in Q\}}(g) 
= f_P(\alpha_g(\omega)),
\]
where
$$
f_P(\omega)=\begin{cases}
       1 &  \text{if }P(v)=\calC_{\omega}(v) \text{ for all }v \in Q ,     \\ 0 &\mbox{otherwise},
      \end{cases}
$$
and hence
\[
\sum_{g\in Q_j\setminus \partial^R Q_j} f_P(\alpha_g(\omega))
\leq
\sharp_P(\calC_\omega|_{Q_j})
\leq 
\sum_{g\in Q_j} f_P(\alpha_g(\omega)).
\]
This gives
\[
 \limsup\limits_{j\rightarrow \infty} \frac{\sharp_P(\calC_\omega|_{Q_j})}{\vert Q_j \vert}
 \leq
 \limsup\limits_{j\rightarrow \infty} \frac{1}{\vert Q_j\vert} \sum\limits_{g\in Q_{j}}f_P(\alpha_g(\omega))
\]
and
\[
 \liminf\limits_{j\rightarrow \infty} \frac{\sharp_P(\calC_\omega|_{Q_j})}{\vert Q_j \vert}
 \geq
\liminf\limits_{j\rightarrow \infty} \frac{1}{\vert Q_j\vert}\! \sum\limits_{g\in Q_j\setminus \partial^R Q_j}\!\!\! f_P(\alpha_g(\omega))
 =
\liminf\limits_{j\rightarrow \infty} \frac{1}{\vert Q_j\vert} \sum\limits_{ g \in Q_{j}}f_P(\alpha_g(\omega)),
\]
where we used that $(Q_j)$ is a F\o lner sequence, cf.\ Remark \ref{remark_Folner}. As $\alpha$ is an ergodic an measure preserving action Theorem \ref{theorem_linde} yields a set $\Omega_P\subseteq\Omega$ of full measure such that the limits
\[
 \lim\limits_{j\rightarrow \infty} \frac{\sharp_P(\calC_\omega|_{Q_j})}{\vert Q_j \vert}
=\lim\limits_{j\rightarrow \infty} \frac{1}{\vert Q_j\vert} \sum\limits_{ g \in Q_{j}}f_P(\alpha_g(\omega))
= \int_\Omega f_P(\omega) d \P(\omega)
\]
exist and are equal for all $\omega \in \Omega_P$. The desired set $\tilde \Omega$ is the (countable) intersection of these $\Omega_P$ for $P\in\bigcup_{j\in\N}\calP(Q_j)$.
\end{proof}

We now focus on the spectral shift function. Since the operators $H_\omega^Q$ are unbounded, the eigenvalue counting functions $n_\omega^Q$ are unbounded as well. However, the spectral shift function for two realisations $H_1^Q$ and $H_2^Q$ with different boundary conditions is bounded, which will be shown in Lemma \ref{lem:xi_est}.

\begin{definition}
	Let $\mathcal{H}$ be a Hilbert space and $H_1, H_2$ be selfadjoint, lower bounded operators with discrete spectra.
	Then the spectral shift function is defined by
	\[\xi_{H_1,H_2}(\lambda) := n_{H_2}(\lambda) - n_{H_1}(\lambda) \quad(\lambda\in\R).\]
\end{definition}

Thus, to obtain properties of $n_{H_2}$ it suffices to study properties of $n_{H_1}$ and $\xi_{H_1,H_2}$.

\begin{lemma}
\label{lem:xi_est}
	Let $H_0$ be a densely defined, closed symmetric and lower bounded operator with deficiency index $k$. Let $H_1$ and $H_2$ be two selfadjoint extensions of $H_0$ with discrete spectrum. Then
	\[\abs{\xi_{H_1,H_2}} \leq k.\]
\end{lemma}

\begin{proof}
	By the min-max principle, for any selfadjoint operator $H$ we have
	\[n_{H}(\lambda) = \max\set{\dim X \mid X\subseteq D(H) \text{ linear subspace},\, H|_X\leq \lambda} \quad(\lambda\in\R),\]
	cf. \cite{glv2007}.
	Now, for $\lambda\in\R$,
	\begin{align*}
		n_{H_2}(\lambda) & = \max\set{\dim X \mid X\subseteq D(H_2) \text{ linear subspace},\, H_2|_X\leq \lambda} \\
		& \leq \max\set{\dim X \mid X\subseteq D(H_0) \text{ linear subspace},\, H_2|_X\leq \lambda} + k \\
		& = \max\set{\dim X \mid X\subseteq D(H_0) \text{ linear subspace},\, H_1|_X\leq \lambda} + k \\
		& \leq \max\set{\dim X \mid X\subseteq D(H_1) \text{ linear subspace},\, H_1|_X\leq \lambda} + k \\
		& = n_{H_1}(\lambda) + k. \qedhere
	\end{align*}
\end{proof}

Changing the boundary conditions of a selfadjoint operator on a graph at one vertex $v$ yields a perturbation of rank at most $2\abs{\calE_v}$. Hence, the spectral shift function of two selfadjoint operators $H_1$ and $H_2$ on a graph which differ only by the boundary conditions at a finite vertex set $Q$ satisfies
\begin{equation}\label{eq:xi_estimate}
	\abs{\xi_{H_1,H_2}} \leq 2\bigcup_{v\in Q} \abs{\calE_v} = 4\abs{Q}\abs{\calS}.
\end{equation}

%

In section \ref{sec:restriction} we defined the eigenvalue counting function $n_\omega^Q$ for the restriction of the 
operator $H_\omega$ to the subgraph $\Gamma_Q$ generated by the set $Q\in \calF$. 
Similarly, we denote the eigenvalue counting function for the Dirichlet Laplacian $-\Delta_D$ restricted 
to $\Gamma_Q$ by $n_D^Q$. The Dirichlet boundary conditions induce that $n_D^Q$ decom\-poses into a sum of 
counting functions, i.\,e.,
\begin{equation}\label{eq_n_DQ}
 n_D^Q(\lambda)= \sum_{e\in \calE_Q} n_{D,s}(\lambda) = |\calE_Q| n_{D,s}(\lambda) = |Q|\sum_{s\in\calS} n_{D,s}(\lambda),
\end{equation}
where $n_{D,s}$ is the eigenvalue counting function of the Dirichlet Laplacian on $L^2(0,l_s)$. We are interested in the spectral shift function 
\begin{equation}\label{def_xi}
 \xi_\omega^Q(\lambda):=n_\omega^Q(\lambda)- n_D^Q(\lambda) = |Q||\calS|\left(N_\omega^Q(\lambda)-\frac{1}{|\calS|}\sum_{s\in\calS}n_{D,s}(\lambda)\right).
\end{equation}
Denote the Banach space of the right-continuous, bounded functions 
$f\from\R\to\R$ equipped with supremum norm by $\calB(\R)$. We study the 
behaviour of the functions $\xi_\omega^{Q_j}$ as $j\to\infty$ as elements of 
$\calB(\R)$, where $(Q_j)$ is a F\o lner sequence.
To this end we prove that $\xi_\omega\from\calF\to\calB(\R)$, 
$Q\mapsto\xi_\omega^Q$ is almost additive and $\calC$-invariant, which will allow for the application of Theorem~\ref{theorem_erg}.
\begin{lemma}
 Let $\omega\in \Omega $ and $\xi_\omega:\calF\to \calB(\R)$, $Q\mapsto\xi_\omega^Q$, where $\xi_\omega^Q$ is given as in \eqref{def_xi}. Then $\xi_\omega$ is almost additive and $\calC_\omega$-invariant.
\end{lemma}

\begin{proof}
	Let $\omega \in\Omega$ and $Q_j\in \calF$, $j=1,\dots,k$ pairwise disjoint be given and set $Q:=\bigcup_{j=1}^k Q_i$. Then
	\begin{align*}
	\left\Vert \xi_\omega^{Q}-\sum_{j=1}^k \xi_\omega^{Q_j} \right\Vert 
	&=  \left\Vert n_\omega^{Q}-n_D^Q - \sum_{j=1}^k (n_\omega^{Q_j}-n_D^{Q_j})
 \right\Vert  \\
	&\leq    \left\Vert n_\omega^{Q} - \sum_{j=1}^k n_\omega^{Q_j} \right\Vert  
	+  \left\Vert n_D^Q - \sum_{j=1}^k n_D^{Q_j} \right\Vert   
	\end{align*}
	holds, where $\Vert\cdot\Vert$ denotes the supremum norm. By \eqref{eq_n_DQ} we have $n_D^Q = \sum_{j=1}^k n_D^{Q_j} $, therefore it remains to prove almost additivity for $n_\omega\from\calF\to \calB(\R), Q\mapsto n_\omega^Q$. 
	Note that $\sum_{j=1}^k n_\omega^{Q_j}$ is the eigenvalue counting function of the operator $\oplus_{j=1}^k H_\omega^{Q_j}$, which equals $H_\omega^Q$ up to the boundary conditions on the vertices $\bigcup_{j=1}^k\calV^{\partial}_{Q_j}$. Now, \eqref{eq:xi_estimate} yields
	\[
	\left\Vert n_\omega^{Q} - \sum_{j=1}^k n_\omega^{Q_j} \right\Vert 
	\leq  4\abs{\calS}\abs{\bigcup_{j=1}^k\calV^{\partial}_{Q_j}}
	\leq  4\abs{\calS} \cdot \sum_{j=1}^k |\partial^1Q_j|
	\]
	which proves almost additivity of $\xi_\omega^Q$ with boundary term $b(Q_j) := 4\abs{\calS}|\partial^1Q_j|$. The $\calC$-invariance of $\xi_\omega$ follows directly from its definition.
\end{proof}

Note that almost additivity and $\calC$-invariance easily imply boundedness, see \cite{lsv2011} for instance.

\begin{corollary}
\label{cor_convergence}
	Let $\calG$ be an amenable group generated by a finite set $\calS$, $\Gamma=(\calV,\calE,\gamma)$ the associated Cayley graph and $(Q_j)$ a tempered F\o lner sequence. Then the limit
	\[
	\lim_{j\to\infty}\frac{\xi_\omega^{Q_j}}{|Q_j||\calS|} 
	\]
	exists in $\calB(\R)$ for almost all $\omega \in \Omega$ and is independent of $\omega$.
\end{corollary}

\section{Proof of main theorem}
\label{sec:proof}

We now prove our main Theorem.

\begin{proof}[Proof of Theorem \ref{thm:main}]
	(i)
	First, we show convergence of $(\abs{\calE_{Q_j}}^{-1} n_\omega^{Q_j})_{j\in\N}$. By Corollary \ref{cor_convergence}, the sequence
	$(\abs{\calE_{Q_j}}^{-1} \xi_\omega^{Q_j})_{j\in\N}$ converges uniformly. Hence, there is $N\from \R\to\R$ such that 
	\[\frac{1}{\abs{\calE_{Q_j}}} n_\omega^{Q_j}= \frac{1}{\abs{\calE_{Q_j}}} \xi_\omega^{Q_j} + \sum_{s\in\calS} n_{D,s} \to N \quad \text{as}\quad j\to\infty\]
	uniformly $\P$-a.\,s. 

	(ii)
	Let $Q\subseteq \calG$ be finite.
	Define $\tilde{N}\from\R\to[0,\infty]$ by
	\begin{align*}
		\tilde{N}(\lambda) & := \frac{1}{\abs{\calE_Q}} \int_\Omega \Tr\left(\1_{\calE_Q} \1_{(-\infty,\lambda]}(H_\omega)\right)\, d\P(\omega).
	\end{align*}
	We show independence of $\tilde{N}$ of the choice of $Q$: by the invariance assumptions, we obtain that
	\[\frac{1}{\abs{\calE_{\{x\}}}}\int_\Omega \Tr\left(\1_{\calE_{\set{x}}} \1_{(-\infty,\lambda]}(H_\omega)\right)\, d\P(\omega) = \frac{1}{\abs{\calS}}\int_\Omega \Tr\left(\1_{\calE_{\set{x}}} \1_{(-\infty,\lambda]}(H_\omega)\right)\, d\P(\omega)\]
	does not depend on $x$. Hence, independence of $Q$ follows.

	(iii)
	We show the equality \eqref{eq:IDS}, i.\,e., $\tilde{N} = N$. Let $\lambda\in\R$ and $Q\subseteq \calG$ be finite.
	Then
	\begin{align*}
		\tilde{N}(\lambda) & = \frac{1}{\abs{\calE_Q}} \int_\Omega \Tr\left(\1_{\calE_Q} \1_{(-\infty,\lambda]}(H_\omega)\right)\, d\P(\omega) \\
		& = \lim_{j\to\infty} \frac{1}{\abs{\calE_{Q_j}}} \int_\Omega \Tr\left(\1_{\calE_{Q_j}} \1_{(-\infty,\lambda]}(H_\omega)\right)\, d\P(\omega) ,\\
	\end{align*}
	and
	\begin{align*}
		N(\lambda) & = \lim_{j\to\infty} \frac{1}{\abs{\calE_{Q_j}}} \Tr\left(\1_{(-\infty,\lambda]}(H_\omega^{Q_j})\right)\\
		& = \int_{\Omega} \lim_{j\to\infty} \frac{1}{\abs{\calE_{Q_j}}} \Tr\left(\1_{(-\infty,\lambda]}(H_\omega^{Q_j})\right)\, d\P(\omega) \\
		& = \lim_{j\to\infty} \frac{1}{\abs{\calE_{Q_j}}} \int_{\Omega}  \Tr\left(\1_{(-\infty,\lambda]}(H_\omega^{Q_j})\right)\, d\P(\omega),
	\end{align*}
	since $\tilde{N}$ does not depend on the choice of $Q$, $\P$ is a probability measure and $N$ is the uniform limit $\P$-a.\,s.

	It suffices to show that the measures associated with $N$ and $\tilde{N}$, respectively, are equal, which in turn follows by vague convergence of the approximating measures $(\mu_j)_{j\in\N}$ and $(\tilde{\mu}_j)_{j\in\N}$, respectively, defined by
	\begin{align*}
		\dupa{f}{\mu_j} & := \frac{1}{\abs{\calE_{Q_j}}} \int_\Omega  \Tr\left(f(H_\omega^{Q_j})\right) \, d\P(\omega), \\
		\dupa{f}{\tilde{\mu}_j} & := \frac{1}{\abs{\calE_{Q_j}}} \int_\Omega \Tr\left(\1_{\calE_{Q_j}} f(H_\omega)\right)\, d\P(\omega), 
	\end{align*}
	for $f\in C_0(\R)$.

	Let us define the set of functions
\[
 \mathcal R:=\set{t\mapsto (t-z)^{-1};\; z\in\C\setminus\R}.
\]
	It is easy to see that the algebra $\alg(\mathcal{R})$ generated by $\mathcal{R}$ separates the points, that it is closed under conjugation and for every $x\in\R$ there is $f\in \alg(\mathcal{R})$ with $f(x)\neq 0$. Thus, the Stone-Weierstra{\ss} Theorem (\cite[Theorem A.10.1]{de2009}) implies that the closure of $\alg(\mathcal{R})$ equals $C_0(\R)$. Moreover, using Cauchy's integral formula, one obtains, that if for all $f\in \mathcal{R}$ one has
	\begin{align}   \label{eq:testfct}
	 \int_\Omega \frac{1}{\abs{\calE_{Q_j}}} \Tr\left(\1_{\calE_{Q_j}} f(H_\omega) - f(H_\omega^{Q_j})\right)\, d\P(\omega) \to 0\quad\text{as}\quad j\to\infty,
	\end{align}
	then this holds true for all $f\in\alg(\mathcal{R})$ as well.
	Hence, in order to prove that \eqref{eq:testfct} holds true for all $f\in C_0(\R)$, it is sufficient to verify this for all $f\in \mathcal{R}$.
	Let $f\in \mathcal{R}$, i.e., $f(t):= (t-z)^{-1}$.
	For $j\in\N$ we can split $\Gamma$ into $\Gamma_{Q_j}$ and $\Gamma_{\calG\setminus Q_j}$. Then $H_\omega$ and
	$H_\omega^{Q_j}\oplus H_\omega^{\calG\setminus Q_j}$ differ only by the boundary conditions on the set $\calV_{Q_j}^\partial$. Thus, by the second resolvent identity,
	\[D:= f(H_\omega) - f(H_\omega^{Q_j}\oplus H_\omega^{\calG\setminus Q_j})\]
	is an operator of rank at most $4\abs{\calS}\abs{\calV_{Q_j}^\partial}$. Moreover, $D$ is bounded by $2\abs{\Im z}^{-1}$, since $f$ is bounded by $\abs{\Im z}^{-1}$. Therefore,
	\begin{align*}
		\abs{\Tr\left(\1_{\calE_{Q_j}} f(H_\omega) - f(H_\omega^{Q_j})\right)} & = \abs{\Tr\left(\1_{\calE_{Q_j}} \left(f(H_\omega) - f(H_\omega^{Q_j}\oplus H_\omega^{\calG\setminus Q_j})\right)\right)} \\
		& \leq \frac{8\abs{\calS}}{\abs{\Im z}} \abs{\calV_{Q_j}^\partial}.
	\end{align*}
	As $(Q_j)$ is a F\o lner sequence, $\calV_{Q_j}^\partial\subseteq \partial^1Q_j$ and $|\calE_{Q_j}|=|\calS||Q_j|$ we obtain
	\[\frac{1}{\abs{\calE_{Q_j}}} \abs{\Tr\left(\1_{\calE_{Q_j}} f(H_\omega) - f(H_\omega^{Q_j})\right)}
		\leq \frac{8\abs{\calS}}{\abs{\Im z}} \frac{\abs{\calV_{Q_j}^\partial}}{\abs{\calE_{Q_j}}} = \frac{8}{\abs{\Im z}} \frac{\abs{\calV_{Q_j}^\partial}}{\abs{Q_j}}
		\to 0\]
	as $j\to\infty$. Since $\P$ is a probability measure, Lebesgue's dominated convergence theorem yields the assertion.
\end{proof}

\section{Application to Heisenberg group}

In the following we discuss the above results in the case where $\calG$ equals to the discrete Heisenberg group $H_3$,
which consists of the elements
\[
(a,b,c):=\begin{pmatrix} 1& 0& 0\\ a &1& 0 \\ c&b&1 \end{pmatrix}, \quad (a,b,c\in\Z).
\]
The group action is induced by the usual matrix multiplication.
$H_3$ is an example of a non-abelian group, which is of polynomial growth.
Therefore it is amenable, as well as residually finite.
One can show, see \cite{lsv2011}, that $H_3$ is generated by $\calS=\{(1,0,0),(0,1,0)\}$ and that $(Q_j)$ given by
\[
Q_j:=\{(a,b,c)\mid 0\leq a,b< j, 0\leq c< j^2 \} \quad (j\in\mathbb N)
\]
is a F\o{}lner sequence. 
We denote the associated metric Cayley graph by $\Gamma=\Gamma(\calG,\calS)=(\calV,\calE,\gamma)$.

Let $(\Omega,\calA,\P)$ be a probability space and let $(H_\omega)_{\omega\in\Omega}$ be
a random Schr\"odinger operator on $\cHg=\bigoplus_{e\in\calE }L^2(0,1)$ defined as in 
\eqref{def:H_omega1} and \eqref{def:H_omega2} (where $l_s=1$ for all $s\in\calS$).

Then Theorem \ref{thm:main} proves that for increasing $j$ 
the eigenvalue counting functions $N_\omega^{Q_j}$ given by
\[N_\omega^Q(\lambda):= \frac{1}{\abs{\calE_Q}}n_\omega^Q(\lambda)
=\frac{1}{\abs{\calE_Q}}|\set{n\in\N \st \lambda_n(H_\omega^Q)\leq \lambda}|
 \quad(\lambda\in \R, Q\in \calF)\]
converge for $\P$-a.\,a.\ $\omega\in\Omega$ uniformly
in the energy variable to the \emph{integrated density of states} $N\from\R\to\R$ defined by
\[
 N(\lambda) := \frac{1}{\abs{\calE_Q}} \int_\Omega \Tr\left(\1_{\calE_Q} \1_{(-\infty,\lambda]}(H_\omega)\right)\, d\P(\omega)\quad(\lambda\in \R),
\]
where $Q\subseteq \calG$ is an arbitrary finite set. Note that for $Q\in\calF$ as usual $(\lambda_n(H_\omega^Q))_{n\in\N}$ is the increasing sequence of eigenvalues of $H_\omega^Q$ counted by multiplicity.

\appendix

\section{Trace class operators on $L(\mathcal{H}_\Gamma)$}

We show that the integral in the Pastur-Shubin formula is finite, i.\,e., that the operator $\1_{\calE_Q}\1_{(-\infty,\lambda]}(H_\omega)$ is trace class for all $\omega\in\Omega$.

Let $H$ be a selfadjoint and semibounded Schr\"odinger operator on $\cHg$ as in section \ref{sec:random_SO}.
Let $Q\subseteq\calG$ be finite. Since $H^Q\oplus H^{\calG\setminus Q}-H$ is of finite rank (they differ only on the boundary conditions at $\calV_{Q}^\partial$), also
\[(H+c)^{-1} - (H^Q\oplus H^{\calG\setminus Q}+c)^{-1} = (H+c)^{-1}(H-H^Q\oplus H^{\calG\setminus Q})(H^Q\oplus H^{\calG\setminus Q}+c)^{-1}\]
has finite rank for sufficiently large $c>0$. Hence,
\[\1_{\calE_Q}((H+c)^{-1} - (H^Q\oplus H^{\calG\setminus Q}+c)^{-1})\]
has finite rank and is therefore trace class.

By \cite[Proposition 5.3 (ii)]{lss2008}, $(H^Q+c)^{-1/2}$ is a continuous 
linear mapping from $\mathcal{H}_{\Gamma_Q}$ to $\bigoplus_{v\in 
Q}\bigoplus_{s\in\calS} L^\infty(0,l_s)$ for sufficiently large $c>0$. Hence, 
by \cite[Satz 6.14]{w2000}, $\1_{\calE_Q} (H^Q+c)^{-1/2}$ is Hilbert-Schmidt. 
But
\[(\1_{\calE_Q} (H^Q+c)^{-1/2})^* = (H^Q+c)^{-1/2} \1_{\calE_Q}\]
is again Hilbert-Schmidt, so
\[\1_{\calE_Q}(H^Q+c)^{-1}\1_{\calE_Q} = \1_{\calE_Q} (H^Q+c)^{-1/2} (H^Q+c)^{-1/2} \1_{\calE_Q}\]
is trace class and therefore also trace class on $\cHg$.

Since
\[\1_{\calE_Q} (H^Q\oplus H^{\calG\setminus Q}+c)^{-1} = \1_{\calE_Q} (H^Q+c)^{-1/2} (H^Q+c)^{-1/2} \1_{\calE_Q},\]
we conclude that
\[\1_{\calE_Q}(H+c)^{-1} = \1_{\calE_Q}((H+c)^{-1} - (H^Q\oplus H^{\calG\setminus Q}+c)^{-1}) + \1_{\calE_Q} (H^Q\oplus H^{\calG\setminus Q}+c)^{-1}\]
is trace class.

Now, we have
\begin{align*}
	\1_{\calE_Q}\1_{(-\infty,\lambda]}(H) & = \1_{\calE_Q}(H+c)^{-1}(H+c)\1_{(-\infty,\lambda]\cap \sigma(H)}(H) \\
	& = \1_{\calE_Q}(H+c)^{-1}\left(z\mapsto (z+c)\1_{(-\infty,\lambda]\cap \sigma(H)}(z)\right)(H).
\end{align*}
Since $\left(z\mapsto (z+c)\1_{(-\infty,\lambda]\cap \sigma(H)}(z)\right)$ is bounded, this operator is trace class as well.
Note that the trace norm depends on $\norm{(H^Q+c)^{-1/2}}_{\mathcal{H}_{\Gamma_Q}\to\bigoplus_{v\in 
Q}\bigoplus_{s\in\calS} L^\infty(0,l_s)}$ and $\norm{\left(z\mapsto (z+c)\1_{(-\infty,\lambda]\cap \sigma(H)}(z)\right)}_\infty$.
Considering now $(H_\omega)_{\omega\in\Omega}$ instead of $H$, these two norms can be bounded uniformly in $\omega$, since there exists $c>0$ such that $H_\omega+c\geq 0$ for all $\omega\in\Omega$.
Since $\P$ is a probability measure this implies the claim.

\section*{Acknowledgements}

The authors thank Daniel Lenz and Ivan Veseli\'c for helpful discussions.

{
}

\bigskip
\noindent
Felix Pogorzelski \\
Friedrich-Schiller-Universit\"at Jena \\
Fakult\"at f\"ur Mathematik und Informatik \\
Ernst-Abbe-Platz 2, 07743 Jena, Germany\\
{\tt felix.pogorzelski@uni-jena.de}\\[3ex]
Fabian Schwarzenberger and Christian Seifert\\
Technische Universit\"at Chemnitz\\
Fakult\"at f\"ur Mathematik\\
09107 Chemnitz, Germany \\
{\tt fabian.schwarzenberger@mathematik.tu-chemnitz.de}\\
{\tt christian.seifert@mathematik.tu-chemnitz.de}

\end{document}